%% file: IV_BEC_wo_GSI_arXiv.tex
\documentclass[journal, draftcls, onecolumn]{IEEEtran}

\usepackage[utf8]{inputenc} 
\usepackage[T1]{fontenc}
\usepackage{url}
\usepackage{ifthen}
\usepackage{cite}
\usepackage[cmex10]{amsmath}

\input{./macro_ihwang.tex}

\usepackage{amssymb,amsfonts,mathtools}
\usepackage{graphicx}
\usepackage{cite}           
\usepackage{booktabs}
\usepackage{array}
\usepackage{balance}        
\usepackage{bm}
\usepackage{bbm}

\usepackage{setspace}
\usepackage{dsfont}
\usepackage{makecell}
\usepackage{xcolor}
\usepackage{enumitem}
\usepackage{verbatim}

\usepackage{tikz}
\usetikzlibrary{positioning,arrows.meta,calc}
\tikzset{
  nodec/.style={draw, circle, minimum size=10mm, inner sep=0pt, font=\tiny},
  chblock/.style={draw, rectangle,
                  minimum height=8mm, minimum width=14mm, align=center, font=\small},
  >={Stealth}
}

\newcommand{\lsepk}{l^{(k)}_{\mathrm{sep}}}

\newcommand{\dsep}{\delta_{\mathrm{sep}}}

\newcommand{\poly}[1]{#1\text{-}\mathrm{poly}}
\newcommand{\linear}[1]{#1\text{-}\mathrm{linear}}
\newcommand{\gsi}{\mathrm{GSI}}

\newtheorem{theorem}{Theorem}[section]

\newtheorem{lemma}[theorem]{Lemma}
\newtheorem{proposition}[theorem]{Proposition}

\newtheorem{definition}[theorem]{Definition}
\newtheorem{remark}[theorem]{Remark}

\newtheorem{assumption}[theorem]{Assumption}

\def \srcpath {./arXiv_v1/}

\allowdisplaybreaks

\begin{document}
\title{On the Information Velocity over a \\Tandem of Erasure Channels}


\author{
  \IEEEauthorblockN{Kai-Chun Chen and I-Hsiang Wang}\\
  \IEEEauthorblockA{Graduate Institute of Communication Engineering\\ 
                    National Taiwan University, Taipei, Taiwan\\
                    Email: \{r13942049, ihwang\}@ntu.edu.tw}
}

\maketitle


\begin{abstract}
\input{\srcpath abstract.tex}
\end{abstract}

\section{Introduction}\label{sec: intro}
\input{\srcpath sec_intro.tex}

\section{Problem Formulation and Preliminaries}\label{sec: formulation}
\input{\srcpath sec_formulation.tex}

\section{Main Results}\label{sec: main_result}
\input{\srcpath sec_results.tex}

\section{Achievability}\label{sec:achi}
\input{\srcpath sec_achievability.tex}


\section{Numerical Comparison}\label{sec: num_comp}
\input{\srcpath sec_numerical.tex}
\bibliographystyle{IEEEtran}
\bibliography{./bib/Ref.bib}

\appendices

\section{Proof of Proposition~\ref{prop: colli}}\label{app:omitted_proofs}
\input{\srcpath app_omitted_proofs.tex}

\section{Converse: Proof of Theorem~\ref{thm: conv_hetero_fb}}\label{sec:conv}
\input{\srcpath sec_converse.tex}

\section{Achievability without GSI for the heterogeneous case}\label{app:het_achieve_woGSI}
\input{\srcpath app_achieve_het_wo_GSI.tex}

\section{Achievability with GSI}\label{app:achieve_w_GSI}
\input{\srcpath app_achieve_w_GSI.tex}


%
%
%
%
%
%
%

\end{document}

%% file: macro_ihwang.tex
\newcommand{\lp}{\left(}
\newcommand{\rp}{\right)}
\newcommand{\lb}{\left[}
\newcommand{\rb}{\right]}
\newcommand{\lbp}{\left\{}
\newcommand{\rbp}{\right\}}
\newcommand{\lba}{\left\lvert}
\newcommand{\rba}{\right\rvert}

\newcommand{\mv}{\,\middle\vert\,}

\newcommand{\mcal}{\mathcal}

\newcommand{\mbb}{\mathbb}

\newcommand{\ra}{\rightarrow}

\newcommand{\eqFunc}{\overset{\mathrm{f}}{=}}
\newcommand{\eqDef}{:=}

\newcommand{\E}{\mathsf{E}}

\newcommand{\W}{\mathsf{W}}

\newcommand{\Var}{\mathsf{Var}}

\renewcommand{\Pr}{\mathsf{P}}

\newcommand{\Indc}[1]{\mathbbm{1}{\lbp #1\rbp}}
\newcommand{\MI}[2]{\mathrm{I}\mkern-1.5mu\left( #1 ; #2\right)}

\newcommand{\CMI}[3]{\mathrm{I}\mkern-1.5mu\left( #1 ; #2  \middle\vert #3\right)}
\newcommand{\ET}[1]{\mathrm{H}\mkern-1.5mu\left( #1\right)}
\newcommand{\CET}[2]{\mathrm{H}\mkern-1.5mu\left( #1  \middle\vert #2\right)}

%% file: arXiv_v1/abstract.tex
Information velocity (IV) is a recently proposed notion to capture the speed of reliable information dissemination over a large-scale network. It is the speed at which reliable end-to-end communication over $k$ hops can be achieved within $t$ time instances, and is defined formally as the asymptotic ratio $k/t$ as $k$ tends to infinity subject to vanishing error probability. To date, even for a tandem of binary erasure channels without feedback, the optimal IV for disseminating multiple (say $m$) bits remains unknown. 
We make progress on this open problem by characterizing the optimal IV for the regime where the message size $m = o(k^{1/2})$. The main contribution lies in achievability, where 
we propose a simple \emph{bit-separation} scheme that pipelines message bits in an orderly fashion with carefully designed temporal spacing so that the flows of different bits do not collide with one another with high probability. This is in sharp contrast to previous attempts in the literature where schemes involve coding over time and across nodes. To go beyond the regime $m = o(k^{1/2})$, we further investigate the setting where every node in the network has strictly causal access to the state information of the BEC links in the entire network. For such a setting with global state information (GSI), we develop an enhanced scheme and characterize the optimal IV for the regime where the message size $m = o(k)$. Interestingly, for the regime $m = o(k^{1/2})$, GSI does not improve the information velocity.

%% file: arXiv_v1/sec_intro.tex
Efficient and timely information dissemination over a large-scale network plays a key role in many applications.
Traditional network information theory has been focused on the notion of network \emph{capacity}, that is, the optimal data rate of reliable communication that can be achieved within a \emph{fixed-size} network in the asymptotic regime where the communication time tends to infinity. Many results have been derived in the literature (see \cite{ElGamalKim_11} for a comprehensive treatment), and in particular, when the network is a tandem of noisy channels, it is straightforward to characterize the network capacity, and it is just the minimum of link capacities in the tandem. 

Meanwhile, such a framework that aims to characterize the network capacity may not well capture the relation between the scale of the communication time and the size of the network.
One of the alternative notions is the \emph{information velocity} (IV) coined in \cite{Huleihel2019}, telling the speed of reliable information dissemination over a large-scale network.
Information velocity is the asymptotic ratio of the ``distance'' (the number of hops from the source to the ``farthest'' node where the message can reliably reach) over the total communication time, as the distance (or time) scales to infinity.
Since a general large-scale network consists of many nodes connected by noisy links, the key to achieving good information velocity is how to judiciously strike the balance between timeliness (latency) and reliability. 
This leads to a new design paradigm of the encoding/decoding functions at the intermediate relay nodes as well as the overall network transmission protocol.

To date, information velocity has been investigated mainly for a tandem of noisy channels. Despite the simplicity of the network topology, as pointed out in \cite{Rajagopalan1994,Huleihel2019}, it is non-trivial to even achieve strictly positive information velocity for general noisy channels. 
The literatures are focused on binary symmetric channels (BSC) \cite{Rajagopalan1994,Huleihel2019,LingScarlett_22}, erasure channels \cite{DomanovitzPhilosof_22,Inovan2024}, and AWGN channels with feedback \cite{Inovan2024,Domanovitz2024}. 
For a tandem of BSC's, Ling and Scarlett \cite{LingScarlett_22} showed that strictly positive IV is achievable. They modified a recursive scheme proposed by Rajagopalan and Schulman in the 90's \cite[Section 6]{Rajagopalan1994} by using a large recursive Hamming code, where different relays perform decode-and-forward at pre-assigned levels. Although a relay may incorrectly decode some coded bits and forward them without detection, the recursive structure ensures that the error probability eventually vanishes. However, the characterization of the optimal IV remains open even for disseminating a single bit, let alone the general case where the message size may scale with the distance (or time). 
For a tandem of AWGN channels, a strictly positive IV was shown to be achievable when every link has feedback from the receiving node to the transmitting node \cite{Inovan2024,Domanovitz2024}. Nevertheless, optimality was not proved, and when there is no feedback, the understanding remains largely open.

The situation becomes more promising when it comes to the case of a tandem of erasure channels. For disseminating a single bit, it is straightforward to achieve the optimal IV by a simple forward-the-last-received scheme. For more than a single bit, Domanovitz \textit{et.al} analyzed a simple Automatic Repeat reQuest (ARQ) scheme assuming that feedback is available and achieved the optimal IV. When there is no feedback, to achieve a strictly positive IV for disseminating more than a single bit, Inovan \cite{Inovan2024} constructed a structured tree code that enables relays to forward message bits in a streaming manner, in order and with controlled reliability. To cope with relays’ uncertainty, the code introduces an additional ``unknown'' branch, allowing a relay to declare that it currently has insufficient information, thereby preventing error amplification. The use of the tree code was inspired by another scheme by Rajagopalan and Schulman \cite[Section 3]{Rajagopalan1994} for BSC's. Unfortunately, the achievable IV does not have a simple form, and numerical evaluation reveals that it is far away from the converse bound derived in \cite{Inovan2024}.

In this work, we make progress on the understanding of information velocity for the tandem of binary erasure channels (BEC). We provide the first characterization of the optimal information velocity for disseminating a general multi-bit message over a tandem of BEC links without feedback. The key idea is to observe that a simple \emph{bit-separation} scheme that pipelines bits in an orderly fashion works well if the timing of sending different bits are carefully spaced so that the flows of different bits do not collide with one another. Our analysis shows that in terms of the scaling of the message size ($m$ bits) with the distance ($k$ hops), the achievable IV of our scheme matches the converse up to the regime $m=\Theta(k^\rho)$, for any $\rho<1/2$. 
The proposed scheme is similar to that in \cite{DomanovitzPhilosof_22} except that neither feedback nor state information of other links is needed. The ``uncoded'' spirit of our scheme is in sharp contrast to the ``coded'' spirit of existing schemes without feedback \cite{Rajagopalan1994,LingScarlett_22,Inovan2024}, where mixing bits over time and across nodes ends up costing additional time to ensure that the nodes are ``on the same page'' for reliable decoding, and as a result, only achieving strictly suboptimal IV. 
We also revisit the scheme of \cite{DomanovitzPhilosof_22} and point out that it seems necessary for each node in the tandem network to have strictly causal access to not only the state information of the succeeding hop (feedback) and the preceding hop, but also that of all the other preceding hops in the entire network. Under such an assumption that \emph{global state information} (GSI) is available, the achievable information velocity is then analyzed by 
leveraging results on last-passage percolation (LPP) \cite{Seppaelaeinen1998, Johansson2000}, which turns out to match the converse up to the regime $m=\Theta(k^\rho)$, for any $\rho<1$. 
LPP is well-studied in probability theory, and it happens to be equivalent to a certain tandem queueing network that appears in our analysis.

Our results indicate that for the regime $m=\Theta(k^\rho)$, GSI does not improve the information velocity for $\rho < 1/2$, that is, the optimal IVs are identical. It is unclear whether or not GSI provides a strict gain in IV for $1/2\le \rho < 1$.

\begin{remark}
Our proposed bit-separation scheme may share a similar flavor with a scheme for a tandem of erasure channels mentioned by Rajagopalan and Schulman in the introduction of their paper \cite{Rajagopalan1994}, and they attributed it to Nisan and Parnas based on the personal communication with N. Nisan \cite[{Ref. [15]}]{Rajagopalan1994}. However, to the best of our knowledge, this ``Nisan-Parnas scheme'' has never been introduced formally and analyzed rigorously in the literature, let alone the investigation on the achievable information velocity. It is even unclear whether or not the scheme requires feedback, and whether or not it is the same as Domanovitz \textit{et.al}'s scheme \cite{DomanovitzPhilosof_22}. In fact, the work by Nisan and Parnas \cite[{Ref. [15]}]{Rajagopalan1994} has not been mentioned in the literatures on information velocity.
\end{remark}



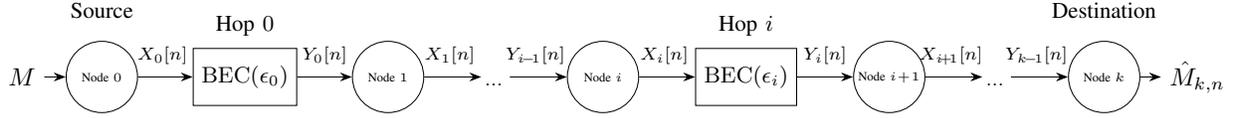
\begin{figure*}[!t]
  \centering
  \resizebox{\textwidth}{!}{%
\begin{tikzpicture}[node distance=18mm]

\node[nodec] (v0) {Node $0$};
\node[nodec, right=30mm of v0] (v1) {Node $1$};
\node[nodec, right=20mm of v1] (vi) {Node $i$};
\node[nodec, right=30mm of vi] (vip1) {Node $i\!+\!1$};
\node[nodec, right=20mm of vip1] (vK) {Node $k$};

\node[left=3mm of v0](Theta) {$M$};
\node at ([xshift=8mm] vK.east) (Thetahat){$\hat{M}_{k,n}$};

\node at ([xshift=10mm,yshift=-1mm] v1.east) {$\scriptstyle\dots$};
\node at ([xshift=10mm,yshift=-1mm] vip1.east) {$\scriptstyle\dots$};

\node[chblock] (c0) at ($(v0)!0.5!(v1)$) {$\mathrm{BEC}(\epsilon_0)$};

\node[chblock] (ci) at ($(vi)!0.5!(vip1)$) {$ \mathrm{BEC}(\epsilon_i)$};

\draw[->] (Theta) -- (v0.west);
\draw[->] (v0) -- (c0.west) node[pos=0.45] (v0c0mid) {};
\draw[->] (c0.east) -- (v1)
node[pos=0.5] (c0v1mid) {};
\draw[->] (v1) -- ([xshift=8mm] v1.east)
node[pos=0.45] (v1c1mid) {};
\draw[->] ([xshift=-8mm] vi.west) -- (vi)
node[pos=0.45] (cim1vimid) {};
\draw[->] (vi) -- (ci.west)
node[pos=0.45] (vicimid) {};
\draw[->] (ci.east) -- (vip1)
node[pos=0.5] (civip1mid) {};
\draw[->] (vip1) -- ([xshift=8mm] vip1.east)
node[pos=0.45] (vip1cip1mid) {};
\draw[->] ([xshift=-8mm] vK.west) -- (vK)
node[pos=0.45] (cKm1vKmid) {};
\draw[->] (vK.east) -- (Thetahat.west);

\node at ([yshift=3mm] v0c0mid) {$\scriptstyle X_0[n]$};
\node at ([yshift=3mm] c0v1mid) {$\scriptstyle Y_0[n]$};
\node at ([yshift=3mm] v1c1mid) {$\scriptstyle X_1[n]$};
\node at ([yshift=3mm] cim1vimid) {$\scriptstyle Y_{i\!-\!1}[n]$};
\node at ([yshift=3mm] vicimid) {$\scriptstyle X_i[n]$};
\node at ([yshift=3mm] civip1mid) {$\scriptstyle Y_i[n]$};
\node at ([yshift=3mm] vip1cip1mid) {$\scriptstyle X_{i\!+\!1}[n]$};
\node at ([yshift=3mm] cKm1vKmid) {$\scriptstyle Y_{k\!-\!1}[n]$};

\node[above=2mm of v0, font=\small] {Source};
\node[above=2mm of vK, font=\small] {Destination};
\node[above=0.5mm of c0, font=\small] {Hop $0$};
\node[above=0.5mm of ci, font=\small] {Hop $i$};

\end{tikzpicture}
}
\caption{A tandem network of BEC links}
  \label{fig:iv_system}
\end{figure*}

%% file: arXiv_v1/sec_formulation.tex
We are asked to communicate a message $M$ uniformly chosen in $\lbp 0,1,\ldots, 2^{m}-1\rbp$ over a tandem network. 
The tandem network consists of $k+1$ nodes, indexed by $\lbp 0,1,\ldots,k\rbp$, and they are connected in tandem by $k$ independent binary erasure channels (BEC's). The $i$-th BEC connects node $i$ and node $i+1$, for $i=0,1,\ldots,k-1$, and its channel law $\W_i(Y_i|X_i)$ (with  channel input $X_i \in \{0,1\}$ and channel output $Y_i \in \{0,1,\ast\}$) is given as follows: 
\begin{equation}
\W_i(y_i|x_i) =
\begin{cases}
1-\epsilon_i, &\text{if $y_i = x_i$}\\
\epsilon_i, &\text{if $y_i = \ast$}.
\end{cases}\ .
\end{equation}
In words, the \emph{erasure probability} of the $i$-th BEC is $\epsilon_i$, and it is assumed that $0<\epsilon_i<1$, $i=0,1,\ldots, k-1$.

The message $M$ is available only at the encoder of the source node (node $0$). 
At each time instance $n \in \mbb{N}$, every node $i\in\lbp 0,1,\ldots,k-1\rbp$ transmits $X_{i}[n]$ to node $i+1$ through the $i$-th BEC $\W_i$, producing an output $Y_i[n]$ observed at node $i+1$. 
For a network consisting of BEC links, one key piece of information at each time instance $n$ is the connectivity of each link, which we defined as the \emph{channel state} in the following:
\begin{equation}
S_i[n] \eqDef \Indc{Y_i[n] \ne \ast} 
=
\begin{cases}
1, &\text{if $Y_i[n] = 0,1$}\\
0, &\text{if $Y_i[n] = \ast$}
\end{cases}\ ,
\end{equation}
indicating whether or not the transmission over the BEC link from node $i$ to node $i+1$ is successful at time $n$.

The transmitted symbol $X_{i}[n]$ from node $i$ at time $n$ is a function of the available information at node $i$ up to time $n-1$. 
We consider setups in two extremes regarding the available information at each node (we use $A \eqFunc B$ to denote ``$A$ is a function of $B$'' and $a[i:j]$ to denote $a_i,a_{i+1},\ldots,a_j$):
\begin{enumerate}
\item Local information at the receiver: $X_i[n]$ only depends on received symbols (and hence the state information of the $(i-1)$-th BEC link) up to time $n-1$, that is,
\begin{equation}
X_0[n] \eqFunc M;\quad
X_i[n] \eqFunc Y_{i-1}[1:n-1],\ i\ge 1.
\end{equation}
\item Global state information at the receiver: $X_i[n]$ depends on received symbols \emph{and} the state information of the \emph{entire network} up to time $n-1$, that is,
\begin{align}
X_0[n] &\eqFunc \big(M, \bm{S}[1:n-1] \big);\\
X_i[n] &\eqFunc \big(Y_{i-1}[1:n-1],  \bm{S}[1:n-1] \big),\ i\ge 1,
\end{align}
where $\bm{S} \eqDef (S_0,S_1,\ldots, S_{k-1})$.

\end{enumerate}
\begin{remark}
We would like to highlight that the relays in the setup with global state information have access to extermely more information about the channels states of the chain than the setting with 1-hop feedback only, where relay $i$ only knows $Y_{i-1}[1:n-1]$ and $S_i[1:n-1]$.
\end{remark}

After $t$ time instances, the destination node (node $k$) produces a decoded output $\hat{M}_{k,t}$. 
As for the performance, we are mainly concerned about the amount of time $t$ needed for the message $M$ to be reliably decoded at the destination node. Let $\xi_k\in (0,1)$ denote the target level of error probability. Let us define the minimum delay as
\begin{align}
\notag&t^*_{m,k}\big(\{\epsilon_i\}_{i=0}^{k-1},\xi_k\big)\\
&\eqDef \inf\lbp t \,\mv\, 
\exists\ \text{encoding/decoding functions}\  
\text{such that}\ \Pr\{\hat{M}_{k,t}\neq M\}<\xi_k
\rbp.
\end{align}



We are interested in the following fundamental question: 
what is the asymptotic limit of the ratio $k/t^*_{m,k}\big(\{\epsilon_i\}_{i=0}^{k-1},\xi_k\big)$,
termed \emph{information velocity}, for reliably disseminating the message over the network as the size of the network $k\ra\infty$? 

Since the minimum delay $t^*_{m,k}\big(\{\epsilon_i\}_{i=0}^{k-1},\xi_k\big)$ also depends on the message size $m$, 
so does the limit of $k/t^*_{m,k}\big(\{\epsilon_i\}_{i=0}^{k-1},\xi_k\big)$ as $k\ra\infty$. 
It turns out that the scaling of the message size with $k$ is particularly interesting. We consider three different regimes:
\begin{enumerate}
\item Constant regime: $m = \Theta(1)$ and does not scale with $k$.
\item Polynomial regime $\poly{\rho}$: $m = \Theta(k^{\rho}),\ \rho \in (0,1)$.
\item Linear regime $\linear{\alpha}$: $m = \alpha k+ o(k),\ \alpha > 0$.
\end{enumerate}

We are now ready to formally define information velocity.
\begin{definition}[Information Velocity]\label{def:IV}
The information velocity over the tandem of independent BEC links with erasure probabilities $\{\epsilon_i \,\vert\, i = 0,1,\ldots\}$ and decoding error probability criteria $\{\xi_i \,\vert\, i = 1,2,\ldots\}$ is defined as
\begin{equation}
v_{\triangle}^{\square}\big( \{\epsilon_i\},\{\xi_i\}\big) \eqDef \liminf_{k\ra\infty} \frac{k}{t^*_{m,k}\big(\{\epsilon_i\}_{i=0}^{k-1},\xi_k\big)},
\end{equation}
where the subscript $\triangle \in \{m,\poly{\rho},\linear{\alpha}\}$ denotes the scaling of the message size with respect to $k$, and the superscript $\square = \gsi$ if global state information is available at all receivers. If only local information is available at all receivers, the superscript $\square$ is simply omitted. Besides, we use a single parameter $\epsilon$ to replace $\{\epsilon_i\}$ when every $\epsilon_i = \epsilon$.
\end{definition}

In the rest of this paper, we shall make the following assumption about the erasure probabilities.
\begin{assumption}\label{assumption:erasure_prob}
$\zeta( \{\epsilon_i\}) \equiv \lim_{k\ra\infty} \frac{1}{k}\sum_{i=0}^{k-1} \frac{1}{1-\epsilon_i}$
exists for the sequence of erasure probabilities $\{\epsilon_i \,\vert\, i = 0,1,2,\ldots\}$. 
Moreover, $\{\epsilon_i\}$ are strictly bounded away from $1$.
\end{assumption}

Note that for homogeneous link erasure probabilities, that is, every $\epsilon_i = \epsilon$, Assumption~\ref{assumption:erasure_prob} holds with $\zeta( \{\epsilon_i\}) = \frac{1}{1-\epsilon}$.

%% file: arXiv_v1/sec_results.tex
Our main contribution is the characterization of the information velocity when receivers only have local information. 

\begin{theorem}[IV without GSI]\label{thm:main}
Suppose that Assumption~\ref{assumption:erasure_prob} holds. 
For the constant message size regime, the information velocity is given as: $\forall\,m\in\mbb{N}$, 
\begin{equation}
v_{m}\lp \lbp\epsilon_i\rbp, \lbp\xi_i\rbp\rp = 1/\zeta( \{\epsilon_i\})
\end{equation}
for the sequence of error probability criteria $\{\xi_i \,\vert\, i = 1,2,\ldots\}$ converging to $0$ from above that satisfy 
\begin{equation}\label{eq:error_criteria_scaling_1}\textstyle
\lim_{k\ra\infty} \tfrac{1}{k} \log(1/ \xi_k) = 0,\ \text{that is, $\xi_k = e^{-o(k)}$.}
\end{equation}
For the polynomial message size regime, the information velocity is given as: $\forall\,\rho\in(0,1/2)$,
\begin{equation}
v_{\poly{\rho}}\lp \lbp\epsilon_i\rbp, \lbp\xi_i\rbp\rp = 1/\zeta( \{\epsilon_i\})
\end{equation}
for the sequence of error probability criteria $\{\xi_i \,\vert\, i = 1,2,\ldots\}$ converging to $0$ from above that satisfy 
\begin{equation}\textstyle
\lim_{k\ra\infty} \tfrac{1}{k^{1-2\rho}} \log(1/ \xi_k) = 0,\ \text{that is, $\xi_k = e^{-o(k^{1-2\rho})}$.}
\end{equation}
%
\end{theorem}
\begin{IEEEproof}[Proof sketch]
The converse is shown by extending arguments in \cite{Inovan2024} (see Theorem~\ref{thm: conv_hetero_fb} below). The achievability is proved by detailed analysis of our proposed scheme, where multiple message bits are pipelined with judiciously designed temporal spacing so that each bit can be disseminated by a simple forward-the-last-received relaying method without interfering with one another. See Section~\ref{sec:achi} for details.
\end{IEEEproof}


Note that in Theorem~\ref{thm:main}, the characterization of IV is only for the regimes of the constant message size and polynomial message size up to $\rho < 1/2$. The reason is that our achievability is based on pipelining message bits over the tandem of BEC links with sufficient spacing, and to ensure that the additional time for all bits to arrive reliably at the final destination is sub-linear in $k$, the scaling of the message size has to be $o(\sqrt{k})$.  

The temporal spacing between consecutive message bits is to ensure that all relay nodes know the status of the network so that they know what bits arrive and what bits to forward. Hence, if global state information (GSI) is available, the spacing can be made much smaller, and as a result, strictly positive IV can be achieved for $\rho \in [1/2,1)$ and even for the linear message size regime. 

\begin{theorem}[IV with GSI]\label{thm: gsi}
For the constant and the polynomial message size regimes, the information velocity is given as: 
$\forall\, m\in\mbb{N}$ and $\forall\,\rho\in(0,1)$, 
\begin{equation}
v^{\gsi}_{m}\lp \epsilon, \lbp\xi_i\rbp\rp 
= v^{\gsi}_{\poly{\rho}}\lp \epsilon, \lbp\xi_i\rbp\rp
= 1-\epsilon.
\end{equation}
for the sequence of error probability criteria $\{\xi_i \,\vert\, i = 1,2,\ldots\}$ converging to $0$ from above that satisfy \eqref{eq:error_criteria_scaling_1}. 
For the linear message size regime, the information velocity is: $\forall\,\alpha>0$, 
\begin{equation}
v^{\gsi}_{\linear{\alpha}}\lp \epsilon, \lbp\xi_i\rbp\rp 
\ge \tfrac{1}{1 + 2\sqrt{\alpha\epsilon}+\alpha}(1-\epsilon).
\end{equation}
for the sequence of error probability criteria $\{\xi_i \,\vert\, i = 1,2,\ldots\}$ converging to $0$ from above that satisfy \eqref{eq:error_criteria_scaling_1}. 
\end{theorem}
\begin{IEEEproof}[Proof sketch]
The converse for the polynomial regime and the constant regime follows Theorem~\ref{thm: conv_hetero_fb}. For achievability, 
we leverage results on last-passage percolation (LPP) \cite{Seppaelaeinen1998, Johansson2000} to extend the analysis of achievable IV to $\poly{\rho}, \rho \ge \frac{1}{2}$ and $\linear{\alpha}, \alpha > 0$. 
See Section~\ref{sec:achi} for details.
\end{IEEEproof}

For the sake of completeness, converse results are given in the theorem below.
\begin{theorem}[Converse]\label{thm: conv_hetero_fb}
Suppose that Assumption~\ref{assumption:erasure_prob} holds. For any $\alpha>0$, $\rho\in(0,1)$, $m\in\mbb{N}$, 
and any $\lbp\xi_i\rbp$ converging to $0$ from above,
	\begin{equation*}
	v_\triangle\lp\lbp \epsilon_i\rbp, \lbp\xi_i\rbp\rp\leq v_\triangle^{\gsi}\lp\lbp \epsilon_i\rbp, \lbp\xi_i\rbp\rp
	\end{equation*}
	for $\triangle\in\lbp 1,m,\poly{\rho},\linear{\alpha}\rbp$. Furthermore,
	\begin{align*}
		&v^{\gsi}_{\linear{\alpha}}\lp\lbp \epsilon_i\rbp, \lbp\xi_i\rbp\rp\\
		&\leq v^{\gsi}_{\poly{\rho}}\lp\lbp \epsilon_i\rbp, \lbp\xi_i\rbp\rp
		\leq v^{\gsi}_m\lp\lbp \epsilon_i\rbp, \lbp\xi_i\rbp\rp \\
		&\leq v^{\gsi}_1\lp\lbp \epsilon_i\rbp, \lbp\xi_i\rbp\rp
		\leq \tfrac{1}{\zeta( \{\epsilon_i\})}.
	\end{align*}
\end{theorem}
See Appendix~\ref{sec:conv} for the proof. 

%% file: arXiv_v1/sec_achievability.tex
In this section, we first describe our scheme with local information only, where multiple message bits are pipelined through the tandem of BEC links with sufficient separation in time. Each (uncoded) message bit is transmitted repeatedly from the source node for some time and disseminated over the tandem, where every relay employs a \emph{forward-the-last-received} scheme. The timing at which the source switches to the next bit should be designed judiciously so that different bits are well separated in time and do not overwrite one another. The analysis of this scheme is carried out to complete the achievability proof of Theorem~\ref{thm:main}. We then switch gears to introduce the enhanced scheme when global state information is available and prove the achievability part of Theorem~\ref{thm: gsi}.

To better convey the ideas, in this section we only present the proofs for the case of homogeneous erasure probabilities. The analysis for the general heterogeneous case can be found in Appendix~\ref{app:het_achieve_woGSI}.

Let us begin with the singe-bit \emph{forward-the-last-received} scheme, which serves as the basic building block of the schemes for multiple bits. 
The forward-the-last-received scheme consists of the following steps.
\begin{itemize}	
	\item The source (node $0$) keeps sending a single-bit message $M=b_0\in\{0,1\}$.
	\item Every relay sends the last non-erased bit that it receives. If there is none, it sends $0$.
	\item The destination (node $k$) declares $\hat{M}$ as the last received bit that is not erased at time $\tau = \frac{k}{1-\epsilon}+\delta$ with $\delta \in o(k)\cap\omega(\sqrt{k})$. 
\end{itemize}
For this simple scheme, the probability of error is upper bounded by the probability that $\tau$ is smaller than the total number of time instances for the single bit to reach the final destination (node $k$), which is the sum of $k$ i.i.d. $\mathrm{Geom}(1-\epsilon)$ random variables. 
Hence, the error probability vanishes as $k\ra\infty$ and an $(1-\epsilon)$ information velocity is achievable.


\subsection{Without GSI: the bit-separation scheme}\label{subsec: scheme}
Towards applying the forward-the-last-received scheme for multiple bits individually, the key is when the source should stop sending the current bit and move on to the next bit. The design of these timings is crucial. If the timings are too close, it may happen that the current  bit successfully arrives at a particular relay while the previous bit has not been forwarded, and consequently the previous bit is lost forever since the relay simply forwards the last bit that is not erased and has no idea the previous bit has never arrived. If the timings are too far way from one another, the overall delay will be lengthened, and the achieved IV will be affected. 

Let us get straight to how the bits are pipelined. Let $M = (b_0,b_1,\ldots,b_{m-1})$ denote the $m$-bit message.
\begin{itemize}
\item 
The source (node $0$) sequentially sends the bit string $b_0,\ldots,b_0,b_1,\ldots,b_1,\ldots,b_{m-1},\ldots,b_{m-1}$, where each bit is repeated for 
$\lsepk= ck^{0.5+\delta_{\mathrm{sep}}}$ 
times. Here $c>0$ and $\dsep>0$ are pre-specified constants. 
We use $\dsep$ to tune the amount of spacing. Later we will see that it is chosen according to the scaling of the message size with $k$. The choice of $c$ does not matter in the regimes of interest.
\item 
Every relay sends the last non-erased bit that it receives. If there is none, it sends $0$.
\item 
The destination (node $k$) declares $\hat{b}_j$ as the last received bit that is not erased at time 
\begin{equation}\label{eq:tau_j}
\tau_{j} = \tfrac{k}{1-\epsilon} + (j+\tfrac{1}{2})\lsepk,\quad j=0,1,\ldots,m-1.
\end{equation}
\end{itemize}




Let us now analyze the scheme for the case of homogeneous erasure probabilities. 
The idea is to 
characterize how far each $b_j$ has successfully reached 
by a Markov process $\{I_{j, n}\}$ starting from time $n=j\lsepk$ at value $0$, that is, $\Pr\{I_{j,j\lsepk}=0\}=1$, and the transition is given as follows: for $n> j\cdot\lsepk$,
\begin{equation}\label{eq:flow_transition}
\Pr\lbp I_{j,n} = i' \,\mv\, I_{j,n-1}=i\rbp
= 
\begin{cases}
1-\epsilon, &i'=i+1\\
\epsilon, &i'=i.
\end{cases}
\end{equation}
One can view $\{I_{j,n}\}$ as the ``wave front'' of the ``flow'' of $b_j$ at time $n$. 
Once for all $j$, the wave front of $b_j$ is not caught up by the next bit $b_{j+1}$ at any given time before the decoding time $\tau_j$ and that $b_j$ is the one read by the destination at time $\tau_j$, all bits can successfully reach the destination in order.
This sufficient condition can be written as: $\forall\, j=0,1,\ldots, m-1$, 
\begin{equation}\label{eq:success_suff_cond_1}
\begin{aligned}
&I_{j+1,n}<I_{j,n}\quad \forall\, n\ \text{with}\ j\lsepk < n < \tau_j\quad \text{and}\\
&I_{j+1,\tau_j}<k \le I_{j,\tau_j},
\end{aligned}
\end{equation} 
where the conditions $I_{m,n}<I_{m-1,n}$ and $I_{m,\tau_{m-1}}<k$ are neglected when $j=m-1$. Let $\mcal{A}_j$ denote the event of the above condition \eqref{eq:success_suff_cond_1}. 
The probability of success of the scheme can then be lower bounded by
$\Pr\lbp\bigcap_{j=0}^{m-1} \mcal{A}_j\rbp$. 

\begin{proposition}\label{prop: colli}
The probability of success of the bit-separation scheme is lowered bounded as
\begin{equation}
\textstyle
\Pr\lbp\bigcap_{j=0}^{m-1} \mcal{A}_j\rbp
\ge
1-2m\exp\lp - c' \frac{k^{2\dsep}}{1+\frac{1}{2}(1-\epsilon)ck^{\dsep-\frac{1}{2}}}\rp
\end{equation}
with $c' = (1-\epsilon)^3c^2/8$.
\end{proposition}

\begin{IEEEproof}[Proof Sketch]
For the event $\bigcap_{j=0}^{m-1}\mcal{A}_j$, instead of directly analyzing its probability, we shall derive a simpler sufficient condition, which says that the wave front of each $\{I_{j,n}\}$ is close enough to its expected value. As long as the expected values are sufficiently separated, $\{I_{j,n}\}$'s do not collide with one another, and hence this is a sufficient condition. Moreover, it is of high probability since the concentration behavior of these processes are guaranteed by martingale concentration. 
Due to space limitation, the details are left in Appendix~\ref{app:omitted_proofs}.
\end{IEEEproof}

As for the achieved IV, note that the overall delay is the decoding time of bit $b_{m-1}$, which is 
\begin{equation}
\tau_{m-1} 
= \tfrac{k}{1-\epsilon} + (m-\tfrac{1}{2})\lsepk
= \tfrac{k}{1-\epsilon} + (m-\tfrac{1}{2})ck^{0.5+\delta_{\mathrm{sep}}}.
\end{equation}
For the constant regime $m=\Theta(1)$, we shall choose $\dsep < 0.5$ so that $\tau_{m-1} = \tfrac{k}{1-\epsilon} + o(k)$, and as a result, the target velocity $1-\epsilon$ is achieved with error probability being $e^{-o(k)}$ by Proposition~\ref{prop: colli}. Similarly, for the polynomial regime $m=\Theta(k^\rho)$, if $0<\rho<1/2$, we can choose $\dsep < 0.5-\rho$ so that $\tau_{m-1} = \tfrac{k}{1-\epsilon} + o(k)$, and as a result, velocity $1-\epsilon$ is achieved with error probability being $e^{-o(k^{1-2\rho})}$ by Proposition~\ref{prop: colli}. This completes the proof of the achievability part of Theorem~\ref{thm:main} for the homogeneous case. 

\begin{remark}[Anywhere reliability]
The proposed bit-separation scheme allows one to track the information flows and ensures that each relay can decode reliably. We leave the treatment in the next version of the paper. 
\end{remark}

%

\subsection{The GSI-control scheme}\label{sec: achieve_with_gsi}
Let us begin with an oversight in \cite{DomanovitzPhilosof_22} that motivates us to revisit their scheme under the assumption that global state information is available.
%
In \cite{DomanovitzPhilosof_22}, it is claimed that, with per-hop feedback, relays employ an ARQ scheme in which each relay repeatedly forwards the oldest bit in its buffer and deletes it once successfully received by its successor. 
The error happens when a relay’s buffer becomes empty (unfortunately not considered in \cite{DomanovitzPhilosof_22}), which occurs at least once with high probability. In that situation, sending an arbitrary $0$ or $1$ confuses the successor, and it will mistake it as one of the message bits. With only ACK/NACK feedback, it is unclear how to explicitly signal that a relay has no bit to send. 


With global state information, the error can be fixed. Since all relays have strictly causal access to the channel states of the entire network, each relay can know when its predecessor’s buffer is empty and ignore the unintended $0$/$1$ transmissions. 

This GSI-control scheme is summarized as follows.
\begin{itemize}
	\item The source (node $0$) starts with a queue containing $b_0,b_1,\ldots,b_{m-1}$ and orderly forwards them. A bit will be deleted from the queue once successfully received.
	\item Every relay maintains a FIFO queue to store the non-erased bits from its preceding hop when its predecessor's queue is not empty. A bit is deleted from the queue once received. When the queue is empty, a $0$ is sent.
	\item The destination (node $k$) stores the non-erased bits with a queue orderly when its predecessor's queue is not empty. It declares the bits in its queue as the decoded bits at a given time (to be specified later).
\end{itemize}

Here we give a sketch of analysis and leave the details in Appendix~\ref{app:achieve_w_GSI}. 
Under this scheme, the flow of message bits behaves as customers in a tandem of queues with independent geometrically distributed service times. Customers experience delay due to the traffic in the network, and this is captured by the \emph{last-passage percolation} (LPP) formulation \cite{Seppaelaeinen1998}. Results in \cite{Johansson2000} are leveraged to provide the analysis of the error probability and the decoding time in both the polynomial regime with $\rho\in(0,1)$ and the linear regime with $\alpha > 0$.
\begin{remark}
One may observe that the GSI-control scheme can be implemented once each relay has access to queue length of the predecessor and the one-bit ACK/NACK feedback from its successor. Therefore, instead of storing the global state information (which is $k\cdot n$ bits in total up to time $n$), each relay $i$ just needs to maintain a counter $\mathrm{Count}_i[n]$ representing its queue length, being updated by the counter $\mathrm{Count}_{i-1}[n-1]$ of its predecessor $i-1$, the incoming symbol $Y_{i-1}[n]$, and the one-bit feedback $S_{i}[n]$ from its successor $i+1$. 
Then, each relay just needs to stores $\log_2 m$ bits, much more memory efficient.
\end{remark}

%% file: arXiv_v1/sec_numerical.tex
Let us numerically compare our results with prior works to close this paper. For comparison, let us focus on the setting without global state information and specialize to the case with homogeneous link erasure probabilities, that is, every $\epsilon_i = \epsilon$. In Figure~\ref{fig:achivs_compare}, for $\epsilon$ ranging from $0$ to $1$, we plot (in log scale) the converse bound of IV (Theorem~\ref{thm: conv_hetero_fb}), along with the achievable IVs of the following schemes:
\begin{enumerate}
\item
Our proposed scheme (matching the converse bound by Theorem~\ref{thm:main}).
\item 
Inovan's scheme \cite{Inovan2024}. The curve plotted here is an analytical lower bound on the achievable IV given in \cite{Inovan2024}. 
\item 
Ling and Scarlett's scheme for BSC's \cite{LingScarlett_22} adapted to BEC's as follows: the inner  code that guarantees nice decaying error probability under BSC's is replaced by sufficient number of repetitions, so that the probability that a random guess is mistaken in the case when all inner coded bits are erased satisfies the flip-probability-requirement in \cite{LingScarlett_22}.
\end{enumerate}
Note that 
for simplicity of comparison among the converse and the achievable IV of the three schemes, the message size $m$ is chosen to be an arbitrary constant, and the error probability criteria $\{ \xi_i\}$ are chosen such that every $\xi_i=\xi$ for some arbitrary $\xi>0$.

From the plot in Figure~\ref{fig:achivs_compare}, one can see that our scheme achieves the converse, from which the achieved IVs of the other two schemes are far away. It is noted for Inovan's scheme \cite{Inovan2024}, even if we take the more complicated form of the achievable IV in \cite{Inovan2024}, there is still a significant gap from the converse. See \cite[Chapter 4]{Inovan2024} for more details.


%

\begin{figure}[htbp]
  \centering
  \includegraphics[width=0.75\textwidth]{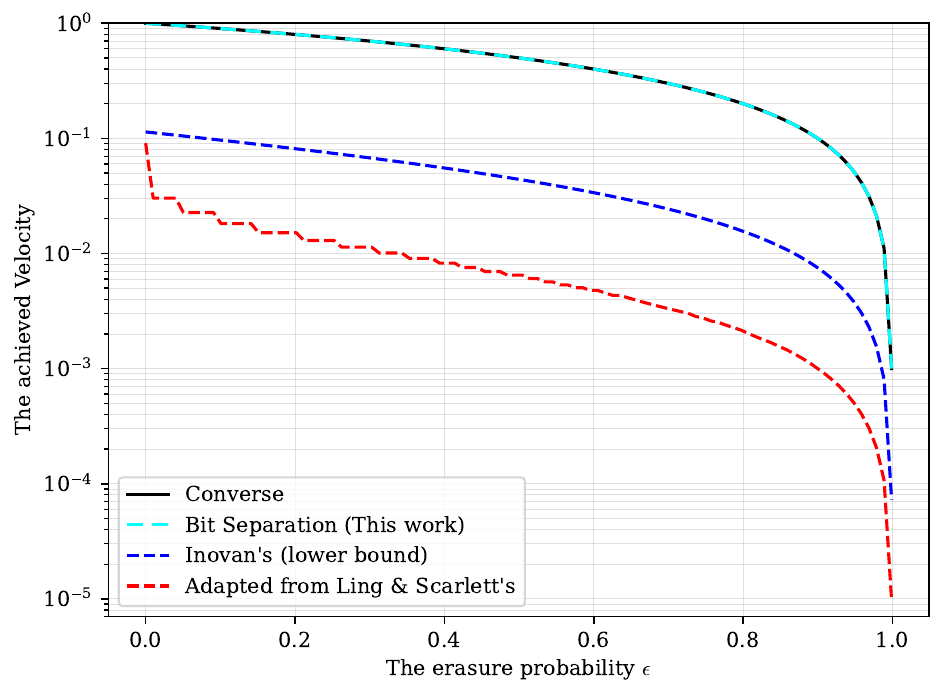}
  \caption{Information velocity (log scale) vs. erasure probability. ``Inovan's'' refers to \cite{Inovan2024}, and ``Ling \& Scarlett's'' refers to \cite{LingScarlett_22}.} 
  \label{fig:achivs_compare}
\end{figure}

%% file: arXiv_v1/app_omitted_proofs.tex
To make the argument in the sketch proof concrete, let us first
derive the expected value of $I_{j,n}$, which can be computed by recursively invoking the law of total expectation: for $j=0,1,\ldots, m-1$ and $n \ge j\lsepk$, 
$\E[I_{j,n}] = (n-j\lsepk)(1-\epsilon) =: \kappa_{j,n}$. 
Note that the separation between the expected values is 
$\kappa_{j,n}-\kappa_{j+1,n}=\lsepk(1-\epsilon)$. 
Also, at the decoding time of $b_j$, we have $\kappa_{j,\tau_j}=k+\frac{1}{2}\lsepk(1-\epsilon)$ and $\kappa_{j+1,\tau_j}=k-\frac{1}{2}\lsepk(1-\epsilon)$ as we plug in $n=\tau_j$ defined in \eqref{eq:tau_j}.

%

Observe that if the following holds $\forall\, j=0,1,\ldots, m-1$, the event $\bigcap_{j=0}^{m-1}\mcal{A}_j$ will be automatically satisfied:
\begin{equation}\label{eq:success_suff_cond_2}
\lba I_{j,n} - \kappa_{j,n}\rba < \tfrac{1}{2}\lsepk(1-\epsilon)\quad \forall\,n\ \text{with}\ j\lsepk < n \le \tau_j.
\end{equation}

The complement of the event of the above condition \eqref{eq:success_suff_cond_2} is
\begin{equation}
\mcal{E}_j \eqDef \lbp \max\limits_{j\lsepk< n\leq \tau_j} \big|I_{j,n} - \kappa_{j,n}\big| \geq \frac{1}{2}\lsepk(1-\epsilon)\rbp,
\end{equation}
which can be understood as the event that the wave front of the flow of $b_j$ \emph{escapes} out of the controlled range. 

To this end, by the union bound, we have 
\begin{equation}\label{eq:union_bd_prop}
\textstyle
\Pr\lbp\bigcap_{j=0}^{m-1} \mcal{A}_j\rbp \ge
\Pr\lbp\bigcap_{j=0}^{m-1} \mcal{E}_j^c\rbp \ge
1 - \sum_{j=0}^{m-1} \Pr\lbp\mcal{E}_j\rbp.
\end{equation}

The following lemma tells that $\Pr\lbp\mcal{E}_j\rbp$ is small.
%
%

\begin{lemma}\label{lem: esc_prob}
For $j=0,1,\ldots,m-1$, 
\begin{equation}\textstyle
\Pr\lbp\mcal{E}_j\rbp \le 2\exp\lp-\frac{\lp\frac{1}{2}\lsepk(1-\epsilon)\rp^2}{2\lp \frac{k}{1-\epsilon} + \frac{1}{2}\lsepk\rp}\rp.
\end{equation}
\end{lemma}
\begin{IEEEproof}[Proof of Lemma~\ref{lem: esc_prob}]
Let $\Delta_{j,n}:= I_{j,n}-\kappa_{j,n}$. By \eqref{eq:flow_transition}, given $\Delta_{j,n},\Delta_{j,n-1},\ldots$, with probability $\epsilon$, 
\[
\Delta_{j,n+1}-\Delta_{j,n} 
= (I_{j,n+1}-I_{j,n}) - (\kappa_{j,n+1} -\kappa_{j,n})  
= -(1-\epsilon),
\]
and with probability $1-\epsilon$, 
\[
\Delta_{j,n+1}-\Delta_{j,n} 
= (I_{j,n+1}-I_{j,n}) - (\kappa_{j,n+1} -\kappa_{j,n})  
= 1-(1-\epsilon) = \epsilon.
\]
As a result, for every bit index $j$, 
\[
\E[\Delta_{j,n+1}-\Delta_{j,n}\,\vert\,\Delta_{j,n},\Delta_{j,n-1},\ldots] = 0, 
\] 
and hence $\{\Delta_{j,n}\}_{n}$ is a martingale. Moreover, it satisfies the bounded difference property: $\forall\,n$, with probability $1$, 
\[
\lba\Delta_{j,n+1}-\Delta_{j,n}\rba = \max(1-\epsilon,\epsilon) \leq 1.
\]
The proof is then complete by applying the maximal Azuma-Hoeffding inequality to the martingale $\{\Delta_{j,n}\}_{n}$ (and also note that the total number of $n$'s is $\tau_j-j\lsepk = \tfrac{k}{1-\epsilon} + \tfrac{1}{2}\lsepk$).
	\end{IEEEproof}
	
The proof of Proposition~\ref{prop: colli} is complete by plugging Lemma~\ref{lem: esc_prob} into \eqref{eq:union_bd_prop}.

%% file: arXiv_v1/sec_converse.tex
%
In the following, we give a proof of Theorem~\ref{thm: conv_hetero_fb}. The procedure of the proof basically follows that in \cite[Section 3]{Inovan2024}. The twist is that we extend it to the setting where relays are allowed to access global state information and the erasure probabilities are heterogeneous along the hops. 
For notational convenience, let us denote the target upper bound  $\tfrac{1}{\zeta( \{\epsilon_i\})}$ by $\tilde{v}_1^{\gsi}$. The first four inequality is straightforward by definition. Therefore, it suffices to prove only the last inequality.

First, we derive the relation of the mutual information of channel outputs and the message over one hop and one time instant.
	\begin{lemma}\label{lem: conv_1hopt1timegap}
Let $\bm{S}^n\equiv \lp S_0[1:n],S_1[1:n],...,S_{k-1}[1:n]\rp$. 
For $i=0$ and all $n\ge 1$,
		\begin{equation}\label{eq: MI_1i1n_ratio_0}
		\frac{\MI{M}{Y_{0}[1:n+1],\bm{S}^{n+1}}-\MI{M}{Y_{0}[1:n],\bm{S}^n}}{1-\MI{M}{Y_{0}[1:n],\bm{S}^n}}\leq 1-\epsilon_{0}.
		\end{equation}
For all $i\ge 1$ and all $n\geq 1$,
		\begin{equation}\label{eq: MI_1i1n_ratio_gen}
		\frac{\MI{M}{Y_{i}[1:n+1],\bm{S}^{n+1}}-\MI{M}{Y_{i}[1:n],\bm{S}^n}}{\MI{M}{Y_{i-1}[1:n],\bm{S}^n}-\MI{M}{Y_{i}[1:n],\bm{S}^n}}\leq 1-\epsilon_{i}.
		\end{equation}
	\end{lemma}
	\begin{IEEEproof}[Proof of Lemma~\ref{lem: conv_1hopt1timegap}]
Let us introduce short-hand notations $C=\lp Y_{i}[1:n],\bm{S}^{n}\rp,\,Z_{\text{in}}=X_{i}[n+1],\,Z_{\text{out}}=Y_{i}[n+1]$. For the denominator, we have that
\begin{align}
		\MI{M}{Y_{i-1}[1:n],\bm{S}^n}-\MI{M}{Y_i[1:n],\bm{S}^n}&=\MI{M}{Y_{i-1}[1:n],X_i[1:n],Z_{\text{in}},\bm{S}^n}-\MI{M}{Y_i[1:n],\bm{S}^n}\\
		&=\MI{M}{Y_{i-1}[1:n],Y_i[1:n],Z_{\text{in}},\bm{S}^n}-\MI{M}{Y_i[1:n],\bm{S}^n}\\
		&\geq \CMI{M}{Z_{\text{in}}}{C},
\end{align}
where the first and second equality is due to Markov chain $M - \lp Y_{i-1}[1:\iota], \bm{S}^{\iota}\rp - X_i[\iota+1] - Y_i[\iota+1]$ for all $\iota\geq 1$.

For the numerator, we have that
		\begin{equation}
		\begin{aligned}
		\MI{M}{Z_{\text{out}},\bm{S}[n+1],C} - \MI{M}{C}&=\CMI{M}{Z_{\text{out}},\bm{S}[n+1]}{C}\\
		&=\CMI{M}{Z_{\text{out}}}{\bm{S}[n+1],C}\\
		&=\Pr\lbp S_{i}[n+1]=0\rbp\cdot \CMI{M}{Z_{\text{out}}}{C,S_i[n+1]=0} \\
		&\quad+ \Pr\lbp S_i[n+1]=1\rbp\cdot \CMI{M}{Z_{\text{out}}}{C,S_i[n+1]=1}\\
		&=\lp1-\epsilon_i\rp\CMI{M}{Z_{\text{in}}}{C},
		\end{aligned}
		\end{equation}
		where the second equality is by the mutual independence among $\bm{S}[n+1]$, $\bm{S}^n$, and $M$, and the third equality is by independence of states between time $n$ and time $n+1$.
		
		As a result, the LHS of the inequality is upper bounded by $1-\epsilon_i$ and \eqref{eq: MI_1i1n_ratio_gen} is proved. \eqref{eq: MI_1i1n_ratio_0} can be shown by replacing $Y_{i-1}[1:n]$ by $M$.
	\end{IEEEproof}
	
With Lemma~\ref{lem: conv_1hopt1timegap}, we introduce the next lemma that bounds the asymptotic behavior of the mutual information between the message $M$ and the observation of a certain relay up to a certain time, where the relay index is linear in the time index.
	\begin{lemma}\label{lem: conv_convergence}
		There exists a function $g:\mbb{Z}_+^2\backslash \lbp\lp0,0\rp\rbp\ra \mbb{R}$ such that
		\begin{equation}\label{eq: MI_bnd_by_g}
		\MI{M}{Y_{i-1}[1:i+n-1],\bm{S}^{i+n}}\leq g(i,n),
		\end{equation}
		for all $i>0$ and that for any $\alpha> \tilde{v}_1^{\gsi}$,
		\begin{equation}\label{eq: g_asymp}
		\lim_{i\ra\infty }g\lp i,\left\lceil\tfrac{1-\alpha}{\alpha}i \right\rceil + 1 \rp = 0.
		\end{equation}
	\end{lemma}
	\begin{IEEEproof}[Proof of Lemma~\ref{lem: conv_convergence}]
	Let $g$ be the function with the following recursive relation:
	\begin{equation}
	g(i,n)=
	\begin{cases}
	0, &n=0\\ 
	1, &i=0\\ 
	\lp 1-\epsilon_{i-1}\rp g(i-1,n) + \epsilon_{i-1} g(i,n-1),&\text{otherwise.}
	\end{cases}
	\end{equation}
	This combined with Lemma~\ref{lem: conv_1hopt1timegap} implies the inequality \eqref{eq: MI_bnd_by_g}. 
	
It is also not hard to see that $g(i,n)$ corresponds to the cumulative distribution function of a sum of random variables: 
\begin{equation}\label{eq: g_to_geosum}
g(i,n)=\Pr\lbp \sum_{\ell=0}^{i-1} G_\ell <i+n\rbp,
\end{equation}
where $\lbp G_\ell\rbp$ denote a sequence of independent geometric random variables and $G_\ell\sim\mathrm{Geom}(1-\epsilon_\ell)$.
	
Let us define $\mu_\ell \eqDef \E\lb G_\ell\rb=\frac{1}{1-\epsilon_\ell}$, $\sigma_\ell^2 \eqDef \Var\lb G_\ell\rb$, $s_i^2 \eqDef \sum_{\ell=0}^{i-1}\sigma_\ell^2 $, and $\hat{G}_\ell\eqDef  \frac{1}{s_i}\lp G_\ell-\mu_\ell\rp$. By the Lyapunov CLT and \eqref{eq: g_to_geosum}, we have the following: for sufficiently large $i$,
\begin{align}
\textstyle
g\lp i,\left\lceil\frac{1-\alpha}{\alpha}i \right\rceil +1 \rp
&\textstyle
\leq\Pr\lbp \sum\limits_{\ell=0}^{i-1} \hat{G}_\ell<\frac{i\lp\frac{1}{\alpha}-\frac{1}{i}\sum_{\ell=0}^{i-1}\frac{1}{1-\epsilon_\ell}\rp+1}{s_i}\rbp\\
&\textstyle
\leq \Pr\lbp \sum\limits_{\ell=0}^{i-1} \hat{G}_\ell<\frac{i\lp\frac{1}{\alpha}-\frac{1}{i}\sum_{\ell=0}^{i-1}\frac{1}{1-\epsilon_\ell}\rp}{\sqrt{i\cdot\lp\max\limits_{0\leq\ell< i}\sigma_\ell+1\rp}}\rbp\\
&\textstyle
\ra 0\quad \text{as $i\ra\infty$}. \qquad\lp\text{ since }\frac{1}{\alpha}<\frac{1}{i}\sum_{\ell=0}^{i-1}\frac{1}{1-\epsilon_\ell}\text{ for large }i\rp
\end{align}
This proves \eqref{eq: g_asymp}
\end{IEEEproof}
	
	By Lemma~\ref{lem: conv_convergence}, $g(i,n)$ is an upper bound of mutual information between $M$ and the observation of relay $i$ after the $(i+n-1)$-th channel use, and when the relay-time ratio is above the threshold $\tilde{v}_1^{\gsi}$, the upper bound converges to 0 in the asymptotic regime.
	Therefore for any $\alpha>\tilde{v}_1^{\gsi}$, by Fano's inequality, 
	\begin{equation}
	\begin{aligned}
	&h_2\lp\Pr\lbp\hat{M}_{i,\left\lceil\frac{i}{\alpha}\right\rceil}\neq M\rbp\rp\\
	&\geq \CET{M}{Y_{i}\lb 1:\left\lceil\frac{i}{\alpha}\right\rceil\rb,\bm{S}^{\left\lceil\frac{i}{\alpha}\right\rceil}}\\
	&= \ET{M} - \MI{M}{Y_{i}\lb 1:\left\lceil\frac{i}{\alpha}\right\rceil\rb,\bm{S}^{\left\lceil\frac{i}{\alpha}\right\rceil}}\\
	&\geq 1-g\lp i, \left\lceil\frac{1-\alpha}{\alpha}i\right\rceil+1\rp\longrightarrow 1\text{ as }i\ra \infty,
	\end{aligned}
	\end{equation}
	where $h_2(\cdot)$ is the binary entropy function.

	Thus, by continuity of $h_2$,
	\begin{equation}
	\Pr\lbp\hat{M}_{i,\left\lceil\frac{i}{\alpha}\right\rceil}\neq M\rbp \ra \frac{1}{2}\text{ as }i\ra\infty.
	\end{equation}

	By Definition~\ref{def:IV}, for any velocity $\alpha > v_1^{\gsi}$, vanishing error probability is impossible, and the proof is completed.

%% file: arXiv_v1/app_achieve_het_wo_GSI.tex
%
In the heterogeneous case, we employ the same scheme as in the homogeneous case, and the approach for performance analysis is similar. In particular, successful dissemination of all the bits is guaranteed as along as the fronts of individual flows of message bits are all concentrated to their respective means. The issue is that, $I_{j,n}-\kappa_{j,n}$ in the proof of Lemma~\ref{lem: esc_prob} (the difference between each front of the flow and its expected value) is no longer a martingale. 

To overcome this technical difficulty, a slight twist to the decoding time is made in the proposed scheme. The decoding time of $\hat{b}_j$ is set to be
\begin{equation}\label{eq: tau_het}
\textstyle
\tau_j=\sum_{i=0}^{k-1}\frac{1}{1-\epsilon_i}+\lp j+\frac{1}{4}\rp\lsepk.
\end{equation}
Slight difference can be observed when \eqref{eq: tau_het} is compared with \eqref{eq:tau_j}. Note that the scale of separation is the same as that in the homogeneous case and hence should be large enough.

Though each hop now is with different erasure probability, one may still characterize 
Since the erasure probabilities are heterogeneous, the transition of the Markov process $\lbp I_{j,n}\rbp$ in \eqref{eq:flow_transition} now becomes
\begin{equation}
\Pr\lbp I_{j,n} = i' \,\mv\, I_{j,n-1}=i\rbp
= 
\begin{cases}
1-\epsilon_i, &i'=i+1\\
\epsilon_i, &i'=i.
\end{cases}
\end{equation}
Note that success of the scheme is still guaranteed by the event $\bigcap_{j=0}^{m-1}\mcal{A}_i$ (see \eqref{eq:success_suff_cond_1}). 
Hence, it suffices to prove the following proposition, which can be viewed as the counterpart of Proposition~\ref{prop: colli} in the homogeneous case.

\begin{proposition}\label{prop: colli_het}
For $k$ large enough,
\begin{equation}\label{eq: suffi_cond_het}
	\Pr\lbp \bigcap_{j=0}^{m-1}\mcal{A}_j\rbp \geq 1-2m\exp\lp -c''\frac{k^{2\dsep}}{1+\frac{c}{4}v_{\min}\max\lp\frac{1}{v_{\min}}-1,1\rp^2 k^{\dsep-\frac{1}{2}}}\rp, 
\end{equation}
where $c''=\frac{c^2}{50}\frac{v_{\min}^3}{v_{\max}^2}$, $v_{\max}=1-\inf_{i\ge0} \epsilon_i$, and $v_{\min}=1-\sup_{i\ge 0}\epsilon_i$. 
Note that by Assumption~\ref{assumption:erasure_prob}, we have $v_{\max}\geq v_{\min}>0$. 
\end{proposition}

\begin{IEEEproof}[Proof of Proposition~\ref{prop: colli_het}]
Instead of $\E[I_{j,n}]$, let us define the following quantity (with slight abuse of notation) to capture the
location of flow $I_{j,n}$:
\begin{equation}\label{eq: kappa_j_het}
	\kappa_{j,n} \eqDef \sup\lbp i\in \mbb{Z}_+\,\bigg| \sum\limits_{\ell=0}^{i-1} \frac{1}{1-\epsilon_\ell}\leq n-j\lsepk\rbp,
\end{equation}
where we set $\epsilon_{i'}=1-v_{\min}\ \forall\,i'\geq k$.

We have the following lemma regarding the spacing between $\kappa_{j,n}$'s for consecutive bits.
\begin{lemma}\label{lem: kappa_sep_hetero}
For $k$ large enough, 
	\begin{equation}
	\lba\kappa_{j,n}-\kappa_{j-1,n}\rba> \frac{1}{2}v_{\min}\lsepk
	\quad \forall\,j=1,2,\dots, m-1,\ \forall\,n>j\lsepk.
	\end{equation}
\end{lemma}
\begin{IEEEproof}[Proof of Lemma~\ref{lem: kappa_sep_hetero}]
	Let $\tilde{i} = \kappa_{j-1,n}$. We have
\begin{align}
	\lba\kappa_{j,n}-\kappa_{j-1,n}\rba &=\left|\sup\lbp i\,\bigg|\, \sum_{\ell=0}^{i-1}\frac{1}{1-\epsilon_\ell}\leq t-j\lsepk\rbp - \tilde{i}\,\right|\\
	&\geq \left|\sup\lbp i\,\bigg| \,\sum_{\ell=i}^{\tilde{i}}\frac{1}{1-\epsilon_\ell}\geq \lsepk\rbp - \tilde{i} \,\right|\\
	&\geq \left| \sup\lbp i\,\bigg|\, (\tilde{i}-i+1)\cdot\frac{1}{v_{\mathrm{min}}} \geq \lsepk \rbp -\tilde{i}\,\right|\\
	&= \lba\left\lfloor\tilde{i} - v_{\mathrm{min}}\lsepk + 1\right\rfloor - \tilde{i}\,\rba\\
	&\geq \frac{1}{2}v_{\mathrm{min}}\lsepk, 
\end{align}
where the last inequality holds when $k$ is sufficiently large.
\end{IEEEproof}

By Lemma~\ref{lem: kappa_sep_hetero} and the fact that 
\begin{equation}
\kappa_{j,\tau_j}\geq k + \frac{1}{4}v_{\min}\lsepk,\ \kappa_{j+1,\tau_j}\leq k-\frac{1}{4}v_{\min}\lsepk,
\qquad \text{(by \eqref{eq: tau_het} and \eqref{eq: kappa_j_het})}
\end{equation}
the event $\bigcap_{j=0}^{m-1}\mcal{A}_j$ holds true 
if the following is satisfied for every $j=0,1,...,m-1$:
\begin{equation}\label{eq:success_suff_cond_2_het}
	\lba I_{j,n}-\kappa_{j,n}\rba < \frac{1}{4}v_{\min}\lsepk\quad\forall \,n\text{ with }j\lsepk<n \leq\tau_j.
\end{equation} 
Let the complement of the event of the above condition \eqref{eq:success_suff_cond_2_het} be
\begin{equation}
	\mcal{E}_j \eqDef \lbp \max_{j\lsepk<n\leq\tau_j}\lba I_{j,n}-\kappa_{j,n}\rba\geq \frac{1}{4}v_{\min}\lsepk\rbp.
\end{equation}
Therefore, again by the union bound, we have
\begin{equation}\label{eq:union_bd_prop_again}\textstyle
\Pr\lbp\bigcap_{j=0}^{m-1} \mcal{A}_j\rbp \ge
\Pr\lbp\bigcap_{j=0}^{m-1} \mcal{E}_j^c\rbp \ge
1 - \sum_{j=0}^{m-1} \Pr\lbp\mcal{E}_j\rbp.
\end{equation}

The following lemma upper bounds the probability of each $\mcal{E}_j$.
\begin{lemma}\label{lem: esc_hetero}
For $k$ large enough,
	\begin{equation}
	\Pr\lbp \mcal{E}_j \rbp\leq 2\exp\lp\frac{-\lp\frac{v_{\min}}{5v_{\max}}\lsepk\rp^2}{2\lp\tau_j-j\lsepk\rp\cdot\max\lp\frac{1}{v_{\min}}-1,1 \rp^2}\rp.
	\end{equation}
\end{lemma}
\begin{IEEEproof}[Proof of Lemma~\ref{lem: esc_hetero}]
Unlike the homogeneous case, here $\lbp I_{j,n}-\kappa_{j,n}\rbp$ is no longer a martingale. Instead, we consider the following term:
\begin{equation}
\Delta_{j,n}:= T(I_{j,n})-\lp n-j\lsepk\rp,\text{ where } T:\mbb{Z}_+\ra\mbb{R}\text{ is defined by }T(i):= \sum_{\ell=0}^{i-1}\frac{1}{1-\epsilon_\ell}.
\end{equation}

By \eqref{eq:flow_transition}, given $\Delta_{j,n},\Delta_{j,n-1},\ldots$, with probability $\epsilon_{I_{j,n}+1}$, 
\begin{equation}
\Delta_{j,n+1} - \Delta_{j,n} 
= \lp T(I_{j,n+1})-\lp \lp n+1\rp-j\lsepk\rp\rp - \lp T(I_{j,n})-\lp n-j\lsepk\rp\rp
= -1,
\end{equation}
and with probability $1-\epsilon_{I_{j,n}+1}$, 
\begin{equation}
\Delta_{j,n+1} - \Delta_{j,n} 
= \lp T(I_{j,n+1})-\lp \lp n+1\rp-j\lsepk\rp\rp - \lp T(I_{j,n})-\lp n-j\lsepk\rp\rp
= \frac{1}{1-\epsilon_{I_{j,n}+1}}-1.
\end{equation}

As a result, for every bit index $j$, 
\[
\E[\Delta_{j,n+1}-\Delta_{j,n}\,\vert\,\Delta_{j,n},\Delta_{j,n-1},\ldots] = 0, 
\] 
and hence $\{\Delta_{j,n}\}_{n}$ is a martingale. Moreover, it satisfies the bounded difference property: $\forall\,n$, with probability $1$, 
\[
\lba\Delta_{j,n+1}-\Delta_{j,n}\rba = \max\lp \frac{1}{1-\epsilon_{I_{j,n}+1}}-1,1\rp 
\le \max\lp\frac{1}{v_{\mathrm{min}}}-1, 1\rp.
\]
%

Moreover, $\forall\,a,b\in\mbb{Z}_+$ with $a<b$,
\begin{equation}\label{eq:bound_diff_T}
	\lba T(b)-T(a)\rba=\lba \sum_{\ell=0}^{b-1}\frac{1}{1-\epsilon_{\ell}}-\sum_{\ell=0}^{a-1}\frac{1}{1-\epsilon_{\ell}}\rba=\lba\sum_{\ell=a}^{b-1}\frac{1}{1-\epsilon_{\ell}}\rba\geq \frac{1}{v_{\max}}\lba b-a\rba.
\end{equation}

Hence, for every $j=0,1,\dots,m-1$ and $\forall\,n>j\lsepk$, we have 
\begin{align}
	\left| I_{j,n} -\kappa_{j,n} \right| &\leq v_{\mathrm{max}}\lba T(I_{j,n})-T(\kappa_{j,n})\rba\\
	&\leq v_{\mathrm{max}}\cdot\max\lbp \lba T(I_{j,n})-\lp n-j\lsepk-\frac{1}{v_{\mathrm{min}}}\rp\rba, \bigg| T(I_{j,n})-\lp n-j\lsepk\rp\bigg| \rbp\\
	&\leq v_{\mathrm{max}}\lba T(I_{j,n}) - (n-j\lsepk)\rba + \frac{v_{\mathrm{max}}}{v_{\mathrm{min}}},
\end{align}
where the first inequality is due to \eqref{eq:bound_diff_T}, and the second inequality holds by the definition of $\kappa_{j,n}$ in \eqref{eq: kappa_j_het}. 

The above inequality leads to a necessary condition of $\mcal{E}_j$:
\begin{align}
&\max\limits_{j\lsepk<n\leq\tau_j}\lba I_{j,n}-\kappa_{j,n}\rba\geq \frac{1}{4}v_{\min}\lsepk \\
\implies 
&\max\limits_{j\lsepk< n\leq \tau_j} \lba T(I_{j,n}) - (n-j\lsepk)\rba \geq \frac{1}{4}\frac{v_{\mathrm{min}}}{v_{\mathrm{max}}}\lsepk - \frac{1}{v_{\mathrm{min}}}. 
\end{align}
Now, since $\lsepk=ck^{0.5+\dsep}$ increases as $k$ grows by our design, the term at the right-hand-side of the last inequality $\frac{1}{4}\frac{v_{\mathrm{min}}}{v_{\mathrm{max}}}\lsepk - \frac{1}{v_{\mathrm{min}}}$ becomes larger than $\frac{1}{5}\frac{v_{\mathrm{min}}}{v_{\mathrm{max}}}\lsepk$ when $k$ is large enough. 
Therefore, for large enough $k$, we can bound the probability of event $\mcal{E}_j$ as follows:
\begin{align}
\Pr\lbp\mcal{E}_j\rbp&= \Pr\lbp \max\limits_{j\lsepk<n\leq\tau_j}\lba I_{j,n}-\kappa_{j,n}\rba\geq \frac{1}{4}v_{\min}\lsepk\rbp\\
&\leq\Pr \lbp \max\limits_{j\cdot\lsepk< n\leq \tau_j} \lba T(I_{j,n}) - (n-j\lsepk)\rba \geq \frac{1}{5}\frac{v_{\mathrm{min}}}{v_{\mathrm{max}}}\lsepk\rbp\\
&\leq 2\exp\lp \frac{-\lp\frac{v_{\mathrm{min}}}{5v_{\mathrm{max}}}\lsepk\rp^2}{2\lp\tau_j-j\lsepk\rp\cdot\max\lp\frac{1}{v_{\mathrm{min}}}-1, 1\rp^2}\rp,
\end{align}
where 
the last inequality is due to the maximal Azuma-Hoeffding inequality.
\end{IEEEproof}

To this end, the proof of Proposition~\ref{prop: colli_het} is complete by Lemma~\ref{lem: esc_hetero} and the union bound \eqref{eq:union_bd_prop_again}.
\end{IEEEproof}



Let us now turn to the analysis of the achieved IV. Note that the overall delay is the decoding time of bit $b_{m-1}$, which is
\begin{equation}
	\tau_{m-1}=\sum_{i=0}^{k-1}\frac{1}{1-\epsilon_i} + \lp \lp m-1\rp+\frac{1}{4}\rp\lsepk=\sum_{i=0}^{k-1}\frac{1}{1-\epsilon_i} + \lp m-\frac{3}{4}\rp ck^{0.5+\dsep}.
\end{equation}

For the constant regime $m=\Theta(1)$, we shall choose $\dsep < 0.5$ so that $\tau_{m-1} = \sum_{i=0}^{k-1}\tfrac{1}{1-\epsilon_i} + o(k)$, and by taking $k\ra\infty$, the target velocity $\frac{1}{\zeta\lp\lbp\epsilon_i\rbp\rp}$ is achieved with error probability being $e^{-o(k)}$ by Proposition~\ref{prop: colli_het}. Similarly, for the polynomial regime $m=\Theta(k^\rho)$, if $0<\rho<1/2$, we can choose $\dsep < 0.5-\rho$ so that $\tau_{m-1} = \sum_{i=0}^{k-1}\tfrac{1}{1-\epsilon_i} + o(k)$, and by taking $k\ra\infty$, velocity $\tfrac{1}{\zeta\lp\lbp\epsilon_i\rbp\rp}$ is achieved with error probability being $e^{-o(k^{1-2\rho})}$ by Proposition~\ref{prop: colli_het}.

%% file: arXiv_v1/app_achieve_w_GSI.tex
We consider the homogeneous case as stated in Theorem \ref{thm: gsi}, where $\epsilon_i=\epsilon$ for all $i$.

Follow the GSI-control scheme we state in Section \ref{sec: achieve_with_gsi}, we now view each message bit as a customer, each hop as a FIFO server whose service time is $\mathrm{Geom}(1-\epsilon)$ with infinite long queue. Say the customers are queued orderly in the server 0 (hop 0), then the total transmition delay is the departure time of the $(m-1)$-th customer ($b_{m-1}$) from the $(k-1)$-th server (hop ($k-1$)), where the servers are cascaded.

Denote the $j$-th customer's departure time from the $i$-th server by $D_{j,i}$. Since all servers follow the FIFO policy, services for the customers are in order anywhere. The service of the $j$-th customer at the $i$-th server starts immediately once that $i$-th server's queue has no former customer than $j$ (the $(j-1)$-th customers has been served at the $i$-th server) and that the $j$-th customer has arrived the $i$-th server (the $j$-th customer has been served at the $(i-1)$-th server). Thus we have the following recursive relation: For any $i,j>0$,
\begin{equation}
\begin{aligned}
D_{0,0}&=S(0,0)\\
D_{0,i}&=D_{0,i-1} + S(0,i)\\
D_{j,0}&=D_{j-1,0} + S(j,0)\\
D_{j,i}&=\max\lbp D_{j-1,i}+D_{j,i-1}\rbp + S(j,i),
\end{aligned}
\end{equation}
where $S(j,i)\overset{\scriptscriptstyle\mathrm{i.i.d.}}{\sim} \mathrm{Geom}(1-\epsilon)$ is the geometric service time of the $j$-th customer at the $i$-th queue.

This by recursion gives that
\begin{equation}
D_{m-1,k-1}=\max_{\begin{subarray}{c}\text{up-right path }p:\\(0,0)\ra(m-1,k-1)\end{subarray}}\sum_{(j,i)\in p} S(j,i)=:T(m,k).
\end{equation}
The term $T\lp\cdot,\cdot\rp$ is exactly the \emph{last-passage-percolation time} with geometric distributed units studied in \cite{Seppaelaeinen1998, Johansson2000}.

Then we introduce the following proposition that tells the large deviation behavior of $T(m,k)$ when $m$ and $k$ are in linear scale to each other.
\begin{proposition}\cite[Theorem 1.1]{Johansson2000}\label{prop: LLP_large_dev}
	For any $\delta,\gamma>0$,
	\begin{equation}
	\lim_{k\ra\infty}-\frac{1}{k}\log\Pr\lbp T\lp\left\lfloor \gamma k\right\rfloor,k \rp \geq k \lp\omega(\epsilon, \gamma) + \delta\rp \rbp= c(\epsilon, \gamma, \delta),
	\end{equation}
	where
	\begin{equation}
	\omega\lp\epsilon, \gamma\rp := \frac{1+2\sqrt{\epsilon\gamma} + \gamma}{1-\epsilon}, \text{ and } c(\epsilon, \gamma, \delta)>0\text{ if } \delta>0.
	\end{equation}
\end{proposition}  

We now turn to analyze the achieved IV. As we state above, the required time for the last message bit to arrive at the destination (node $k$) is exactly $D_{m-1,k-1}=T(m,k)$. 

Therefore, for the polynomial regime where $m=\Theta\lp k^{\rho}\rp$, we set the decoding time simply as $k\lp\frac{1}{1-\epsilon}+\delta'\rp$ for arbitrary $\delta'>0$. Then, by applying Proposition \ref{prop: LLP_large_dev} with $\gamma=\Theta(k^{\rho-1})$, any IV below $1-\epsilon$ is achievable for error probability with exponential decay in $k$. Hence,
\begin{equation}
	v^{\gsi}_{\poly{\rho}}\lp \epsilon, \lbp\xi_i\rbp\rp \geq 1-\epsilon, \text{ for any }\lbp\xi_i\rbp \text{ such that }\xi_k=e^{-o(k)}.
\end{equation}

Correspondinly, for the linear regime where $m=\alpha k + o(k)$, we set the decoding time as $k\lp\omega(\epsilon, \alpha)+\delta''\rp$, for arbitrary $\delta''>0$. Then, by applying Proposition \ref{prop: LLP_large_dev} with $\gamma=\alpha$, any IV below $1/\omega\lp\epsilon,\alpha\rp$ is achievable for error probability with exponential decay in $k$. Hence,
\begin{equation}
	v^{\gsi}_{\linear{\alpha}}\lp\epsilon,\lbp\xi_i\rbp\rp\geq \frac{1}{\omega\lp\epsilon,\alpha\rp} =\frac{1-\epsilon}{1 + 2\sqrt{\epsilon\alpha}+\alpha},\text{ for any }\lbp\xi_i\rbp \text{ such that }\xi_k=e^{-o(k)}.
\end{equation}

This completes the achievability part of Theorem \ref{thm: gsi}.

%% file: IV_BEC_wo_GSI_arXiv.bbl
\begin{thebibliography}{1}
\providecommand{\url}[1]{#1}
\csname url@samestyle\endcsname
\providecommand{\newblock}{\relax}
\providecommand{\bibinfo}[2]{#2}
\providecommand{\BIBentrySTDinterwordspacing}{\spaceskip=0pt\relax}
\providecommand{\BIBentryALTinterwordstretchfactor}{4}
\providecommand{\BIBentryALTinterwordspacing}{\spaceskip=\fontdimen2\font plus
\BIBentryALTinterwordstretchfactor\fontdimen3\font minus
  \fontdimen4\font\relax}
\providecommand{\BIBforeignlanguage}[2]{{%
\expandafter\ifx\csname l@#1\endcsname\relax
\typeout{** WARNING: IEEEtran.bst: No hyphenation pattern has been}%
\typeout{** loaded for the language `#1'. Using the pattern for}%
\typeout{** the default language instead.}%
\else
\language=\csname l@#1\endcsname
\fi
#2}}
\providecommand{\BIBdecl}{\relax}
\BIBdecl

\bibitem{ElGamalKim_11}
A.~E. Gamal and Y.-H. Kim, \emph{Newtork Information Theory}.\hskip 1em plus
  0.5em minus 0.4em\relax Cambridge University Press, 2011.

\bibitem{Huleihel2019}
W.~Huleihel, Y.~Polyanskiy, and O.~Shayevitz, ``Relaying one bit across a
  tandem of binary-symmetric channels,'' in \emph{2019 IEEE International
  Symposium on Information Theory (ISIT)}, 2019, pp. 2928--2932.

\bibitem{Rajagopalan1994}
\BIBentryALTinterwordspacing
S.~Rajagopalan and L.~Schulman, ``A coding theorem for distributed
  computation,'' in \emph{Proceedings of the Twenty-Sixth Annual ACM Symposium
  on Theory of Computing}, ser. STOC '94.\hskip 1em plus 0.5em minus
  0.4em\relax New York, NY, USA: Association for Computing Machinery, 1994, pp.
  790--799. [Online]. Available: \url{https://doi.org/10.1145/195058.195462}
\BIBentrySTDinterwordspacing

\bibitem{LingScarlett_22}
Y.~H. Ling and J.~Scarlett, ``Simple coding techniques for many-hop relaying,''
  \emph{IEEE Transactions on Information Theory}, vol.~68, no.~11, pp.
  7043--7053, 2022.

\bibitem{DomanovitzPhilosof_22}
E.~Domanovitz, T.~Philosof, and A.~Khina, ``The information velocity of
  packet-erasure links,'' in \emph{IEEE INFOCOM 2022 - IEEE Conference on
  Computer Communications}, no. 190-199, 2022.

\bibitem{Inovan2024}
\BIBentryALTinterwordspacing
R.~Inovan, ``\BIBforeignlanguage{en}{On speed and advantage: Results in
  information velocity and monitoring problems},'' Ph.D. dissertation, EPFL,
  Lausanne, 2024. [Online]. Available:
  \url{https://infoscience.epfl.ch/handle/20.500.14299/207871}
\BIBentrySTDinterwordspacing

\bibitem{Domanovitz2024}
E.~Domanovitz, A.~Khina, T.~Philosof, and Y.~Kochman, ``Information velocity of
  cascaded {G}aussian channels with feedback,'' \emph{IEEE Journal on Selected
  Areas in Information Theory}, vol.~5, pp. 554--569, 2024.

\bibitem{Seppaelaeinen1998}
T.~Sepp{\"a}l{\"a}inen, ``Hydrodynamic scaling, convex duality, and asymptotic
  shapes of growth models,'' \emph{Markov Processes and Related Fields}, 01
  1998.

\bibitem{Johansson2000}
K.~Johansson, ``Shape fluctuations and random matrices,'' \emph{Communications
  in Mathematical Physics}, vol. 209, no.~2, pp. 437--476, Feb. 2000.

\end{thebibliography}
